\documentclass{article}

\usepackage{arxiv}

\usepackage[utf8]{inputenc} 
\usepackage[T1]{fontenc}    
\usepackage{hyperref}       
\usepackage{url}            
\usepackage{booktabs}       
\usepackage{amsfonts}       
\usepackage{nicefrac}       
\usepackage{microtype}      
\usepackage{lipsum}		
\usepackage{graphicx}
\newtheorem{theorem}{Theorem}

\newtheorem{proposition}[theorem]{Proposition}
\newtheorem{remark}[theorem]{Remark}

\newenvironment{proof}[1][Proof]{\noindent\textbf{#1.} }{\ \rule{0.5em}{0.5em}}

\title{ Stochastic modeling using Adomian method and Fractionnal differential
equations.}

\author{\href{https://orcid.org/0000-0000-0000-0000}{\includegraphics[scale=0.06]{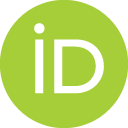}\hspace{1mm}Komla Elom ADEDJE} \\
	Université Thomas SANKARA\\
	BP: 417 Ouagadougou 12, Burkina Faso.\\
	\texttt{keadedje@gmail.com} \\
	\And	
	\href{https://orcid.org/0000-0000-0000-0000}{\includegraphics[scale=0.06]{orcid.png}\hspace{1mm}Diakarya Barro} \\
	Université Thomas SANKARA \\
	BP: 417 Ouagadougou 12, Burkina Faso.\\
	\texttt{Corresponding author : dbarro2@gmail.com} \\
}

\begin{document}
\maketitle

\begin{abstract}
In this paper, we propose a fractional differential equation of order one-half, to model the evolution through time of the dynamics of accumulation and elimination of the contaminant in human organism with a deficient immune system, during consecutive intakes of contaminated food. This process quantifies the exposure to toxins of subjects living with comorbidity (children not breast-fed, the elderly, pregnant women) to food-born diseases. The Adomian Decomposition Method and the fractional integration of Riemann Liouville are used in the modeling processes.
\end{abstract}


\section{Introduction}
\section{Introduction}

Food safety is an important issue in the search for solutions to the various pathologies such as food-borne diseases, cancer, renal failure, diabetes, AIDS, Crohn's disease, ulcerative colitis, etc that the human species has been facing in recent decades. The continuous time model named Kinetic Dietary Exposure Model (KDEM) introduced in \cite{verger2007integration} and later in \cite{tillier2014theorie}, is a stochastic process which allows to model the phenomena of accumulation and elimination of the contaminant from the human organism and which gives an answer to this quest for a solution to the current health problems. Their model, aims to represent the evolution through time of a quantity of  contaminant in the human body.\\

Note that this theory developed to model the evolution of contaminant dynamics in the body, does not take into account a failing immune system. Subjects who have contracted a
bacteria, malnourished children, diabetics etc, have a weak immune system. So, the speed of elimination of the contaminant in the body becomes slow. Note that the linear pharmacokinetic model with one compartment is used in the existing model to describe the phenomenon of elimination of the contaminant by the immune system. The quantity $ X\left(t\right) $ of the contaminant in the body between two consecutive intakes decreases according to the linear differential equation of the form  
\begin{equation}\label{myeq1}
\frac{dx\left( t\right) }{dt} = - \theta x\left( t\right); 
\end{equation}
where $ X\left( t\right) $ is given time t, $ \theta > 0 $ is the removal rate of the contaminant between the dates $ T_{i} $ and $ T_{i+1}, $ with $ \theta = ln2/DV, $ DV being the biological half-life of the contaminant (see more in \cite{tressou2004probabilistic}).\\

The dynamics of the evolution of the accumulation and elimination of the contaminant proposed in the KDEM model is of the form 

\begin{equation}\label{myeq2}
X_{n+1}\left( t\right)  = X_{n}.e^{-\theta t} + E_{n+1};
\end{equation}
and the exposure E of an individual is calculated for all p = 1,...,P as follows  

\begin{equation}\label{myq3}
E = \frac{\displaystyle\sum_{p=1}^{P}Q_{p}C_{p}}{w},
\end{equation}
where P products, $ Q_{p} $ contamination of p foods, $ C_{p} $ consumption of p foods, $ w $ the individual's body weight (see \cite{clemenccon2009exposition}). \\
The structural properties of this exposure process have been studied extensively in \cite{debes2006impact} using integrated Markov chain analysis (see more in \cite{spatial}). Such autoregressive models with random coefficients have been widely studied in the literature.  \\

In contrast to KDEM model, we propose in this paper an approach based on fractionnal calculations, leading to a process of the dynamics of the evolution of the accumulation and elimination of the contaminant. A process that takes into account the delay in the elimination of the contaminant when the immune system is weakened by a pathology or a situation of malnutrition in children and pregnant women. However, we propose a FDE  that models this delay in elimination. Let's recall for this purpose, that in biology, it was deduced that the membranes of cells of biological organism have electrical conductance of fractional order \cite{chen2005rlc} and then is classified in group of models not - integer.  \\

The main contribution of this paper is to propose a dynamic model that takes into account the accumulation and elimination of contaminants from the human body of subjects living with comorbidity. In this contex, the exposure process is described by a random variable X, which represents the quantity of contaminant ingested during a short period.\\
 The rest of the paper is organized as follows. In section \ref{sec:headings1}, we will first lay the mathematical foundations to build the stochastic process that will be used to dynamically model the evolution of the accumulation and elimination of a contaminant when the human organ system fails. Section \ref{sec:headings2} proposes a stochastic model called fractional differential equation (FDE) of order 1/2 whose particularity is to take into account the delay in the evolution of a phenomenon. Moreover, we make a brief comparison of the existing classical KDEM model that we implemented with the dioxin dietary intake data and the first model that we elaborated with the failed immune systems.

\section{Materials and methods}
\label{sec:headings1}
In this section, we summarise the fractional integration of Riemann-Liouville and the Adomian decomposition method. These mathematical tools turn out to be very important for our study.

\subsection{ An overview of fractional integral of Riemann-Liouville}
\label{part1}
Successive integer-order derivatives have been used to model many life phenomena, including exposure to food risks. Here, our approach is based on fractional integral of Riemann-Liouville (see \cite{abdeljawad2016geometric} and \cite{loufouilou2013application} for more details).\\
 Let $ \alpha \in \mathbb{R}_{+}^{*}, a \in \mathbb{R} $ and f a locally integrable function defined on $ \left[ a, +\infty\right),$ the integral of order $ \alpha $ of $f$ of lower bound a is defined for all $ t\geq \alpha, $ such as :  
	
	\begin{equation}\label{my4}
	_{a}I_{t}^{\alpha}f\left( t\right) =\frac{1}{\Gamma\left( \alpha\right) }\int_{a}^{t}\left( t - \tau\right)^{\alpha-1}f\left( \tau\right)d\tau;
	\end{equation}
while $ \Gamma\left( x\right) = \int^{+\infty}_{0} t^{x-1}e^{-t}dt$ is the well known Euler gamma function. This integral is convergent, for all $ x > 0. $ Note that the Beta function is related to the Gamma one and plays an important role in fractional calculations which is defined by  \cite{fridjet2018derivee} such as, for all x, y $>$ 0, 
\begin{equation}\label{my2}
 B\left( x, y\right) = \int_{0}^{1}t^{x-1}\left( 1 - t\right) ^{y-1}dt = \frac{\Gamma\left( x\right) \Gamma\left( y\right) }{\Gamma\left( x + y\right) }.
\end{equation} 
Furthermore, the fractional left derivative of order $ \alpha $ over a given interval $\left[a, b \right]$ of Riemann Liouville is such as, for all $\alpha > 0, x \in \left[a, b \right]$ and n = $\left[\alpha \right] + 1$ 
\begin{equation}
\mathcal{D}_{a^{+}}^{\alpha}f\left(x\right)=\frac{1}{\Gamma\left( n - \alpha\right) }\frac{d^{n}}{dx^{n}}\left(\int_{a}^{x}\left( x - t\right)^{n - \alpha - 1 }f\left(t\right)dt\right); 
\end{equation}
where $\left[\alpha \right]$ is the integer part of $ \alpha $.\\
\subsection{An overview of Adomian Decomposition Method} 
\label{part2}
Suppose that we need to solve a functional equation of the form
\begin{equation}\label{eq}
Au = f; \hspace{0.2cm} \mbox{where} \hspace{0.2cm} A : H \rightarrow H
\end{equation}
is a nonlinear operator defined on a real Hilbert space H, u is the unknown function defined in H and f is a function given in H. First, we decompose the operator A into :
\begin{center}
	A = L + R + N ;
\end{center}
where L is an operator inversible in the sense of Adomian, R is the rest of the linear part and N is the the nonlinear part. So, $\left( \ref{eq}\right)$ provides the relation

\begin{equation}
 Lu + Ru + Nu = f.
\end{equation}
The operator L is inversible, we obtain the Adomian canonical form \cite{yindoula2014application}

\begin{equation}
u = \theta + L^{-1}f - L^{-1}Ru - L^{-1}Nu;
\end{equation}
where $ \theta $ is a constant such that : $ L\theta = 0. $
In particular, the solution of equation $\left( \ref{eq}\right)$ is sought in the form of a series and we decompose the nonlinear part into :

\begin{equation}
Nu = \displaystyle\sum_{n = 0}^{+\infty}A_{k}\left( u_{0},...,u_{k}\right).
\end{equation}
Furthermore, one has :

\begin{equation}\label{sum}
\displaystyle\sum_{n = 0}^{+\infty}u_{n} = \theta + L^{-1}f - L^{-1}R\left( \displaystyle\sum_{n = 0}^{+\infty}u_{n}\right) - L^{-1}N\left( \displaystyle\sum_{n = 0}^{+\infty}u_{n}\right).
\end{equation}
We obtain by identification of the terms the following algorithm under the hypothesis of the convergence of the series $ \left( \displaystyle\sum_{n = 0}^{+\infty}u_{n}\right);$  for all $n \geq 0 $, see \cite{beyi2020application}
\begin{equation}
\left\lbrace \begin{array}{ll}
u_{0} = \theta + L^{-1}f \\
u_{n + 1} = -L^{-1}R\left( u_{n}\right) - L^{-1}N\left( u_{n}\right); 
\end{array}\right.
\end{equation}
The nonlinear part is obtained from the Adomian polynomials defined as follows
 
\begin{equation}
A_{n} = \frac{1}{n!}\left[ \frac{d^{n}}{d\lambda^{n}}N\left( \displaystyle\sum_{i=1}^{+\infty}\lambda^{i}u_{i}\right) \right] _{\lambda = 0}.
\end{equation}
In practice, the following formula is used to calculate polynomials  (see \cite{el2007new} for more details). For all n $\geq 0$, one has
\begin{equation}
\left\lbrace \begin{array}{ll}
A_{0} = N\left( u_{0}\right)  \\
A_{n + 1} = \frac{1}{n + 1}\left( \displaystyle\sum_{k=0}^{n}\left( k + 1\right)u_{k+1}\frac{\partial A_{n}}{\partial u_{k}}\right). 
\end{array}\right.
\end{equation}
The following section deals with the main of our paper results.

\section{Main results}
\label{sec:headings2}
 In this paper we propose an alternative method to the existing ones by the non-integer order Riemann-Liouville derivation, an integration that is the inverse operation of the derivation, which effectively accounts for the delay in the elimination of the contaminant from the body by the immune system during the evolution of the phenomenon. Before building the model, we let's make the following assumptions :\\
 
\noindent
(C1): The subjects living with co-morbidities have a less efficient immune system and therefore fail in contrast to immunocompetent subjects.\\

\noindent
(C2): The amount of food toxin in the body today depends on the amount accumulated yesterday, thus dependence between contaminations in time. \\

\noindent
Let's notice that when the human organism is immunocompromised, the rate of elimination of the contaminant from the body by the immune system is slow compared to an immunocompetent organism. Hence, the need for a fractional differential equation that takes into account the delay in the evolution of the phenomenon. Let's begin by this example where we present the difference between the non-integer derivation and the fractional derivation in the sense of Riemann Liouville of function is that the non-integer one undergoes a delay in its evolution. 
	\begin{equation}
	D_{0^{+}}^{1/2}x=\frac{1}{\Gamma\left( \frac{1}{2}\right) }\frac{d}{dx}\left(\int_{0}^{x}\left( x - t\right) ^{-\frac{1}{2}}tdt\right); 
	\end{equation}
where $D_{0^{+}}^{1/2}x$ is the variation in the amount of contaminant in the body. Now, by using the well known relationship, the change of variable $ dt = xds, s \in \left[ 0,1\right] $
	
It comes that,
	\begin{equation}
 D_{0^{+}}^{1/2}x  =\frac{1}{\Gamma\left( \frac{1}{2}\right) }\frac{d}{dx}\int_{0}^{1}\left( x - sx\right) ^{-\frac{1}{2}}sx.xds; 
	\end{equation}
and furthermore, after a simple calculation, one obtains
	 \begin{equation}
	D_{0^{+}}^{1/2}x  = \frac{3}{2}x^{\frac{1}{2}}\frac{1}{\Gamma\left( \frac{1}{2}\right) }\frac{\Gamma\left( \frac{1}{2}\right) \Gamma\left( 2\right) }{\Gamma\left( \frac{5}{2}\right) }= \frac{2\sqrt{x}}{\sqrt{\pi}}.
	\end{equation}
	
\begin{remark}
A relatively delayed stochastic process can be modeled by a fractional differential equation.
\end{remark}

\subsection{ Univariate stochastic process }
Note that when a system is disturbed or delayed, it allows a slowed motion, hence the need to model its trajectory by fractional derivation. So, subjects living with comorbidity develop pathology more quickly following exposure to food toxins. This is due to the slow elimination of toxins from their immune systems. The delay in the elimination of contaminants ingested during dietary intake in comorbid individuals can be well modeled by a fractional differential equation of order 1/2. Indeed, the linear one-compartment pharmacokinetic model long used by toxicologists is that of an ordinary differential equation of order 1

\begin{proposition}\label{myth}
	
Let A be the initial body burden of contaminant at date $ T_{0}= 0.$ Between two consecutive intakes the amount of contaminant in the immunodeficient organism decreases with time according to the following fractional differential equation of non-integer order :
	
	\begin{equation}
	\left\lbrace \begin{array}{ll}
	D^{1/2} X\left( t\right)= - \theta X\left( t\right)  \\
	X\left( 0\right) = A
	\end{array}\right.
	\end{equation} 
The exposure computed immediately after the i-th food intake is such as
\begin{equation}
\label{propo}
	X\left( t\right) = A\left( e^{\theta^{2} t} - \frac{4\mid \theta \mid\sqrt{t}}{\sqrt{\pi}}\displaystyle\sum_{n = 0}^{+\infty}\frac{4^{n}\left( n + 1\right) ! \left( \theta^{2}t\right) ^{n}}{\left( 2n + 2 \right) !}\right) ;
\end{equation}
where $ X\left( t\right) $ is the amount of the contaminant in the body at given time t, $ \theta > 0 $ is the removal rate of the contaminant between the dates $ T_{i} $ et $ T_{i+1}, $ where $ \theta = ln2/DV, $ DV being the biological half-life of the contaminant; X(0) is the initial body burden of contaminant at $ T_{0} = 0. $

\end{proposition}
The human organism has a system of accumulation in contaminants and these last ones being eliminated in a progressive way and also depending on the frequency of the food intake.\\
\begin{proof}
\noindent	
~We propose a proof based on fractional Riemann Liouville integration and the Adomian decomposition method.\\

\noindent
Consider the following equation,	
\begin{equation}\label{myeq2}
   \left\lbrace \begin{array}{l} 
	D^{1/2} X\left( t\right) + \theta X\left( t\right) = 0 \\
	X\left( 0\right) = A 
	\end{array}\right. 
\end{equation}
Then, it follows that
	\begin{equation}
	\left\lbrace \begin{array}{lll} 
	LX & = & D_{0^{+}}^{1/2} X \\
	RX & = & \theta X\left( t\right) \\
	g & = & 0 
	\end{array}\right. \hspace{0.2cm} \mbox{where} \hspace{0.2cm} L^{-1} = I_{0^{+}}^{1/2}.
	\end{equation}
 By using $\left(\ref{myeq2}\right)$, it comes that 
	\begin{equation}\label{myeq3}
	LX + RX = g .
	\end{equation}
By using $ L^{-1} $ to $\left(\ref{myeq3}\right),$ one has :

\begin{equation}
 L^{-1} \left( LX\right) = I_{0^{+}}^{1/2}\left( D_{0^{+}}^{1/2} X\left( t\right) \right) = X\left( t\right) - X\left( 0\right).
\end{equation} 
 It comes that
	\begin{equation}\label{myeq5}
  X\left( t\right) - X\left( 0\right) +  L^{-1}\left( RX\right) = L^{-1} g. 
	\end{equation}
That allows us to consider	with the form
	

	\begin{equation}\label{myeq6}
	 X = \displaystyle\sum_{n=0}^{+\infty} X_{n} .
	\end{equation}
From (\ref{myeq6}) and (\ref{myeq5}), it yields that
 
	\begin{equation}
	\Sigma_{n=0}^{+\infty} X_{n}\left( t\right) = X\left( 0\right) + L^{-1}g - \Sigma_{n=0}^{+\infty} L^{-1}\left( RX_{n}\left( t\right) \right); 
	\end{equation}
giving the following system
	\begin{equation}
	\left\lbrace \begin{array}{ll} 
	X_{0}\left( t\right)  = X\left( 0 \right) + L^{-1}g \\
	X_{n+1}\left( t\right)  = - L^{-1}\left( R\left( X_{n}\left( t\right) \right) \right) & n\geq 0 .
	\end{array}\right.
	\end{equation}
Calculation of the terms of $ X_{n}\left( t\right). $ \\
At order 0, we have $X_{0}\left( t\right)$ such as :

	\begin{equation}\label{myeq7}
	X_{0}\left( t\right) = X\left( 0\right) + L^{-1}g 
	= A .
	\end{equation}
A order 1, $X_{1}\left(t\right),$ using $\left(\ref{myeq7}\right)$ one has :

	\begin{equation}\label{myeq8}
	 X_{1}\left( t\right) = 
	 =  -I_{0^{+}}^{1/2}\left( \theta  X_{0}\left( t\right)\right) 
	 = - I_{0^{+}}^{1/2}\left( \theta A\right) = - \frac{2 \theta A t^{1/2}}{\sqrt{\pi}}.
	\end{equation}
By using $\left(\ref{myeq8}\right),$  one has :	

	\begin{equation} \label{myeq10}
 X_{2}\left( t\right) =  - L^{-1}\left( R X_{1}\left( t\right) \right)  
	  =  - I_{0^{+}}^{1/2}\left( \theta\left( \frac{-2\theta A t^{1/2}}{\sqrt{\pi}} \right) \right) 
    =  \frac{2\theta^{2} A I_{0^{+}}^{1/2}t^{1/2}}{\sqrt{\pi}}.
	\end{equation}	
From $\left( \ref{eqs1}\right)$ and $\left( \ref{my2}\right),$ it comes that
	\begin{equation} \label{myeq11}
	 I_{0^{+}}^{1/2}\left( t^{1/2}\right)  =  \frac{1}{\Gamma\left( \frac{1}{2}\right) }\int_{0}^{1}\left( t - st\right) ^{-1/2}\left( st\right) ^{1/2}tds = \frac{t}{\Gamma\left( \frac{1}{2}\right) }\int_{0}^{1}\left( 1 - s\right) ^{\frac{1}{2}-1}s^{1/2}ds = \frac{t\sqrt{\pi}}{2}
	 \end{equation}
By introducing (\ref{myeq11}) into (\ref{myeq10}), it comes that 
	\begin{equation} \label{my1}
	X_{2}\left( t\right)  =  \frac{2\theta^{2} A}{\sqrt{\pi}}\frac{t\sqrt{\pi}}{2} 
	  =  \theta^{2} A t .
	\end{equation}
	For $X_{3}\left(t\right),$ one has
	\begin{equation} 
   X_{3}\left( t\right)  =  - L^{-1}\left( R X_{2}\left( t\right) \right) .
   \end{equation}
   By replacing (\ref{my1}) in the above formula, it comes that
   \begin{equation}\label{myeq12}
	 X_{3}\left( t\right)  =  -I_{0^{+}}^{1/2}\left( \theta\left( \theta^{2} A t\right) \right) 
      = -\theta^{3} A I_{0^{+}}^{1/2}t =\frac{t^{1/2}}{\Gamma\left( \frac{1}{2}\right) }\int_{0}^{1}\left( 1 - s\right)^{-1/2}sds  =  \frac{-4\theta^{3} At^{3/2}}{3\sqrt{\pi}}.
    \end{equation}
In an analogous way one obtains from successively the terms

	\begin{eqnarray}
	\left\lbrace \begin{array}{lllllll}
	X_{0}\left( t\right)  = A \\
	X_{1}\left( t\right)  = - \frac{2 \theta A t^{1/2}}{\sqrt{\pi}} \\
	X_{2}\left( t\right) = \theta^{2} A t \\
	X_{3}\left( t\right) =  \frac{-4\theta^{3} At^{3/2}}{3\sqrt{\pi}} \\
	X_{4}\left( t\right) = \frac{\theta^{4}At^{2}}{2} \\
	X_{5}\left( t\right) = -
	\frac{8\theta^{5}At^{5/2}}{15\sqrt{\pi}} \\
	X_{6}\left( t\right) = \frac{\theta^{6}At^{3}}{6} \\
	X_{7}\left( t\right) = -\frac{16\theta^{7}At^{7/2}}{105\sqrt{\pi}} \\
	\vdots  
	\end{array}\right. 
	\end{eqnarray}	
However in pratice, all terms of the series cannot be determined; so we use an approximation of the solution from the truncated series as follows	
	\begin{equation}
	X\left( t\right)   =  \displaystyle \sum_{n=0}^{+\infty}  X_{n}\left( t\right); 
	\end{equation}
which provides the following equivalences
\begin{eqnarray}
X\left( t\right) & \simeq &  X_{0}\left( t\right) + X_{1}\left( t\right) + X_{2}\left( t\right) + X_{3}\left( t\right) + X_{4}\left( t\right) + X_{5}\left( t\right) + X_{6}\left( t\right) + X_{7}\left( t\right) \cdots \\ 
	& \simeq & A - \frac{2\theta At^{1/2}}{\sqrt{\pi}} + \theta^{2}At - \frac{4\theta^{3}At^{3/2}}{3\sqrt{\pi}} + \frac{\theta^{4}At^{2}}{2} - \frac{8\theta^{5}At^{5/2}}{15\sqrt{\pi}} + \frac{\theta^{6}At^{3}}{6} - \frac{16\theta^{7}At^{7/2}}{105\sqrt{\pi}} + \cdots\\
	& \simeq & A\left( 1 + \theta^{2}t + \frac{\theta^{4}t^{2}}{2!}  + \frac{\theta^{6}t^{3}}{3!} + \cdots \right) + \\
	& + &\frac{A}{\sqrt{\pi}}\left( -\frac{\theta t^{1/2}}{\frac{1}{2}} - \frac{\theta^{3}t^{3/2}}{\frac{3}{2}\times\frac{1}{2}} - \frac{\theta^{5}t^{5/2}}{\frac{5}{2}\times\frac{3}{2}\times\frac{1}{2}} - \frac{\theta^{7}t^{7/2}}{\frac{7}{2}\times\frac{5}{2}\times\frac{3}{2}\times\frac{1}{2}} + \cdots \right). \\ 
	\end{eqnarray}
	In the same way, it follows that
	\begin{eqnarray}
X\left( t\right) & \simeq & A\displaystyle \sum_{n=0}^{+\infty}\frac{\left( \theta^{2}t\right) ^{n}}{n!} - A\displaystyle \sum_{n=0}^{+\infty}\frac{\left( \theta^{2}t\right) ^{\frac{1}{2}+n}}{\Gamma\left( \frac{1}{2}+ n + 1\right) } \\ 
	& \simeq & Ae^{\theta^{2} t}- A\displaystyle \sum_{n=0}^{+\infty}\frac{\left( \theta^{2}t\right) ^{\frac{1}{2}+n}}{\Gamma\left( \frac{1}{2}+ n + 1\right) } \\ 
	& \simeq & Ae^{\theta^{2} t} - A\frac{4\sqrt{\theta^{2} t}}{\sqrt{\pi}}\displaystyle\sum_{n = 0}^{+\infty}\frac{4^{n}\left( n + 1\right) ! \left( \theta^{2}t\right) ^{n}}{\left( 2n + 2 \right) !} \\ 
	& \simeq & Ae^{\theta^{2} t} - A\frac{4\mid \theta \mid\sqrt{t}}{\sqrt{\pi}}\displaystyle\sum_{n = 0}^{+\infty}\frac{4^{n}\left( n + 1\right) ! \left( \theta^{2}t\right) ^{n}}{\left( 2n + 2 \right) !}\cdot \\ 
\end{eqnarray}	
Finally, one obtains
\begin{equation}\label{lig}
	X\left( t\right) \simeq  A\left( e^{\theta^{2} t} - \frac{4\mid \theta \mid\sqrt{t}}{\sqrt{\pi}}\displaystyle\sum_{n = 0}^{+\infty}\frac{4^{n}\left( n + 1\right) ! \left( \theta^{2}t\right) ^{n}}{\left( 2n + 2 \right) !}\right)
\end{equation}	
	
\end{proof}

Note that $ X\left( t\right) $ represents the behavior of the quantity of contaminated intake at each consumption in humain body through time for subjects living with comorbidity.

\begin{proposition}	
Under (C1), the trajectory of the accumulated exposure is defined as follows
	\begin{equation}\label{tra}
	X_{i+1}\left( t\right) = X_{i}\left( e^{\theta^{2} t} - \frac{4\mid \theta \mid\sqrt{t}}{\sqrt{\pi}}\displaystyle\sum_{n = 0}^{+\infty}\frac{4^{n}\left( n + 1\right) ! \left( \theta^{2}t\right) ^{n}}{\left( 2n + 2 \right) !}\right) + E_{i+1};
	\end{equation}
$ E_{i+1} $ is the calculated exposure at time $ t_{i+1}. $ 
	
\end{proposition}

\begin{proof}
~According to proposition $\left(\ref{myth}\right)$ the first term of $\left(\ref{tra}\right)$ is proved. We obtain the second term proposed in the work of Bertail in \cite{clemenccon2009exposition} and \cite{debes2006impact}, i.e the exposure over a given period is none other than the sum of the consumption of the contaminated products, weighted by the contamination rates associated with each of the products of doubtful quality. By noting P the number of foodstuffs carrying the contamination $C_{P}$, the consumption of any individual of body weight w and $Q_{P}$ the contamination rate in $\mu g/kg$ or pg/kg of each of these foodstuffs. This leads to the result of the random exposure of an individual to the contaminant of interest as follows $\left( \ref{myq3}\right) $ concluding the proof
\end{proof}
\subsection{Conditional dependence of multivariate stochastic process }
Let us make some additional assumptions\\

(C3): Let us define the health risk arising from the accumulation and interaction between the different contaminants in the body. For instance, this involves assessing the risk due to the accumulation of three substances in the body. It should be noted that studies have shown that there is a dependency between the trajectories of contaminants in the body.\\

(C4): The body burden of each contaminant during the month depends on the body burden of that contaminant in the previous month and the new exposure.\\

Based on the fact that contaminants are stored and eliminated by the same organ, we make the following dependency hypothesis. The similar approach can be found in \cite{analysis}.  The above assumptions lead to the following proposition 

\begin{proposition}
Let $ X_{i} = \left( X_{i1},X_{i2},X_{i3}\right), i=1,...,n $ be the vectors body burden of each contaminant during $i^{th}$ week or month. Then, the evolution of the quantity of contamiant over time, is getting by the following relation 
\begin{equation}
\left( X_{i}\right) = \left( X_{i-1}\right)\psi + E_{i}
\end{equation}
 with $ \psi=\mathcal{M}_{3\left( \mathbb{R}\right)} $ the matrix of transition probabilities given by \\
 
$ \left( \begin{array}{ccc}
 e^{\theta_{1}^{2} t} - \frac{4\mid \theta_{1} \mid\sqrt{t}}{\sqrt{\pi}}\displaystyle\sum_{n = 0}^{+\infty}\frac{4^{n}\left( n + 1\right) ! \left( \theta_{1}^{2}t\right) ^{n}}{\left( 2n + 2 \right) !} &  a & b \\
c & e^{\theta_{2}^{2} t} - \frac{4\mid \theta_{2} \mid\sqrt{t}}{\sqrt{\pi}}\displaystyle\sum_{n = 0}^{+\infty}\frac{4^{n}\left( n + 1\right) ! \left( \theta_{2}^{2}t\right) ^{n}}{\left( 2n + 2 \right) !} & d \\
e & f & e^{\theta_{3}^{2} t} - \frac{4\mid \theta_{3} \mid\sqrt{t}}{\sqrt{\pi}}\displaystyle\sum_{n = 0}^{+\infty}\frac{4^{n}\left( n + 1\right) ! \left( \theta_{3}^{2}t\right) ^{n}}{\left( 2n + 2 \right) !}
\end{array}\right) $

where $ \theta_{1}, \theta_{2}, \theta_{3} $ denote the removal rates of the contaminant in the body and $E_{i} = \left( E_{i1},E_{i2},E_{i3}\right)$ the vectors of the trajectories of the exposures to the contaminants during the $i^{th}$ week or month. The matrix $ \psi $ contains contaminant removal coefficients.

\end{proposition}

\begin{proof} ~The proof of the previous proposition based on assumptions (C3) and (C4), yield the desired result
\end{proof}

\noindent
It can be seen easily that
$$ \Longrightarrow \left\lbrace \begin{array}{lll}
X_{i1}  & = &  X_{i-1,1}\ast  \left( e^{\theta_{1}^{2} t} - \frac{4\mid \theta_{1} \mid\sqrt{t}}{\sqrt{\pi}}\displaystyle\sum_{n = 0}^{+\infty}\frac{4^{n}\left( n + 1\right) ! \left( \theta_{1}^{2}t\right) ^{n}}{\left( 2n + 2 \right) !}\right)  + c\ast X_{i-1,2}+ e\ast X_{i-1,3} + E_{i1}  \\
 X_{i2}  & = & a\ast X_{i-1,1} +  X_{i-1,2}\ast \left(  e^{\theta_{2}^{2} t} - \frac{4\mid \theta_{2} \mid\sqrt{t}}{\sqrt{\pi}}\displaystyle\sum_{n = 0}^{+\infty}\frac{4^{n}\left( n + 1\right) ! \left( \theta_{2}^{2}t\right) ^{n}}{\left( 2n + 2 \right) !}\right) + f\ast X_{i-1,3} + E_{i2}  \\
 X_{i3} & = & b\ast X_{i-1,1} + d\ast X_{i-1,2} +  X_{i-1,3}\ast  \left( e^{\theta_{3}^{2} t} - \frac{4\mid \theta_{3} \mid\sqrt{t}}{\sqrt{\pi}}\displaystyle\sum_{n = 0}^{+\infty}\frac{4^{n}\left( n + 1\right) ! \left( \theta_{3}^{2}t\right) ^{n}}{\left( 2n + 2 \right) !}\right) + E_{i3} 
\end{array}\right. $$
Neglecting the impact of the removal of one contaminant on the other, we have trajectory $\left( \ref{tra}\right) $ for each contaminant considered, i.e. the coefficients a, b, c, and d are zero. However, when the contaminants are stored in the same organ, then the elimination become more slow, because of the dependency of the elimination. Assumptions will be made to determine the reals a, b, c, d, e, f of the matrix by an Archimedean copula dependence investigate in \cite{barro2016spatial}, \cite{barro} and estimate the probability that the amount of contaminant 1 in the body exceeds a threshold u1 knowing that
contaminant 2 has exceeded the threshold u2. Unfortunately this will not be examined in this work.

\subsection{Generalized Fractional Differential Equation}
Assume once again that the manifestation of a pathology differs from one person to another. Based on this assumption, we propose a generalized fractional differential equation of order $ \alpha $ which takes into account the delay in the evolution of a phenomenon.
\begin{proposition}
Let the following system provides a fractional differantial equation 
\begin{equation}
   \left\lbrace \begin{array}{l} 
	D^{\alpha} X\left( t\right) + \theta X\left( t\right) = 0 \\
	X\left( 0\right) = \kappa 
	\end{array}\right.
\end{equation}
such that t $ > $ 0; $0 < \alpha \leq 1 $;
modeling the evolution through time of the dynamics of accumulation and elimination of the contaminant in human organism with a deficient immune system, during consecutive intakes of contaminanted food. Then, its the solution given by   : 
 \begin{equation}
X\left( t\right) = \displaystyle \sum_{n=0}^{+\infty}\frac{\left( -1\right)^{n}\theta^{n}\kappa t^{n\alpha}}{\Gamma\left( n\alpha + 1\right)};
\end{equation}
is a stochastic time varying process which quantifies the exposure to toxins of subjects living with comorbidity.
\end{proposition}

\begin{proof}
~The proof requires frequent uses of parts  $ \left( \ref{part1}\right) $ and $ \left( \ref{part2}\right) $ of section \ref{sec:headings1}. Indeed, it is easy to show that 

\begin{equation}
 X_{0}\left( t\right) = X\left( 0\right) = \kappa.
\end{equation}
A order 1, the quantity  $  X_{1}\left( t\right), $ gives 
\begin{equation}
\label{eqm}
 X_{1}\left( t\right) = -L^{-1}\left( R X_{0}\left( t\right)\right) = -\left( \frac{\theta \kappa  t^{\alpha}}{\Gamma\left( \alpha + 1\right) }\right).
\end{equation}
By introducing $\left( \ref{eqm}\right)$ into the following relation, one has
\begin{equation} 
\label{eqm3}
 X_{2}\left( t\right) = -L^{-1}\left( R X_{1}\left( t\right) \right) =  \frac{\theta^{2} \kappa I_{0^{+}}^{\alpha}\left(  t^{\alpha}\right)}{\Gamma\left( \alpha + 1\right) }.
\end{equation}

So, we get the following results :
\begin{equation} 
\label{eqm2}
 I_{0^{+}}^{\alpha}\left(  t^{\alpha}\right)=\frac{1}{\Gamma\left( \alpha\right) }\int_{0}^{1}\left( t - st\right)^{\alpha-1}\left( st\right)^{\alpha} tds =  \frac{t^{2\alpha}}{\Gamma\left( \alpha\right)}\beta\left( \alpha, \alpha + 1\right) =  \frac{t^{2\alpha}\Gamma\left( \alpha + 1\right)}{\Gamma\left( 2\alpha + 1\right)}.
\end{equation}
Futhermore, substituting $ \left(\ref{eqm2}\right) $ into $ \left(\ref{eqm3}\right), $ one has
\begin{equation}
X_{2}\left( t\right) = \frac{\theta^{2}\kappa}{\Gamma\left( \alpha + 1\right)}\times \frac{t^{2\alpha}\Gamma\left( \alpha + 1\right)}{\Gamma\left( 2\alpha + 1\right)} = \frac{\theta^{2}\kappa t^{2\alpha}}{\Gamma\left( 2\alpha + 1\right)}.
\end{equation}
For $  X_{3}\left( t\right),  $ we have
\begin{equation}
 X_{3}\left( t\right) = -L^{-1}\left( R X_{2}\left( t\right) \right) = -I_{0^{+}}^{\alpha}\left( \frac{\theta^{3}\kappa t^{2\alpha}}{\Gamma\left( 2\alpha + 1\right)}\right) = \frac{-\theta^{3}\kappa}{\Gamma\left( 2\alpha + 1\right)}I_{0^{+}}^{\alpha}\left( t^{2\alpha}\right).
\end{equation}
Applying the same operator, it comes that
\begin{equation}
I_{0^{+}}^{\alpha}\left(  t^{2\alpha}\right) = \frac{1}{\Gamma\left( \alpha\right)}\int_{0}^{1} t^{3\alpha}\left( 1 - s\right)^{\alpha-1} s^{2\alpha} ds = \frac{t^{3\alpha}}{\Gamma\left( \alpha\right)}\beta\left( \alpha, 2\alpha + 1\right).
\end{equation}
And finally, it comes that
\begin{equation}
 X_{3}\left( t\right) = \frac{-\theta^{3}\kappa}{\Gamma\left( 2\alpha + 1\right)}\frac{t^{3\alpha}\Gamma\left( 2\alpha + 1\right)}{\Gamma\left( 3\alpha + 1\right)} = \frac{-\theta^{3}\kappa t^{3\alpha}}{\Gamma\left( 3\alpha + 1\right)}.
\end{equation}
By conjecture, we have the solution of the problem $\left( \ref{myeq2}\right) $ as follows
\begin{equation}
X_{n}\left( t\right) = \frac{\left( -1\right)^{n}\theta^{n}\kappa t^{n\alpha}}{\Gamma\left( n\alpha + 1\right)}.
\end{equation}
According to $\left(\ref{sum}\right)$ we finally obtain the following result
\begin{equation}
X\left( t\right)  = \displaystyle \sum_{n=0}^{+\infty}  X_{n}\left( t\right) = \displaystyle \sum_{n=0}^{+\infty}\frac{\left( -1\right)^{n}\theta^{n}\kappa t^{n\alpha}}{\Gamma\left( n\alpha + 1\right)}.
\end{equation}
Our result is proved if we show by recurrence that 
\begin{equation}
X\left( t\right) = \frac{\left( -1\right)^{n}\theta^{n}\kappa t^{n\alpha}}{\Gamma\left( n\alpha + 1\right)}.
\end{equation}

\begin{equation}
\mbox{For the order 0,} \hspace{0.5cm} n = 0, \hspace{0.5cm} X\left( t\right) = X_{0}\left( t\right) = \kappa.
\end{equation}
The property is true at order $n=0$, now let us suppose it true at any order $n$ and show that the property is true at order $n+1$, i.e 

\begin{equation}
X_{n + 1}\left( t\right) = \frac{\left( -1\right)^{n + 1}\theta^{n + 1}\kappa t^{\left( n + 1\right)\alpha}}{\Gamma\left( \left( n+1\right)\alpha + 1\right)}.
\end{equation}
Using the same approach, we obtain the results successively as follows
\begin{equation}
X_{n + 1}\left( t\right) = - L^{-1}\left( RX_{n}\left( t\right) \right) = -I_{0^{+}}^{\alpha}\left( \theta\left( \frac{\left( -1\right)^{n}\theta^{n}\kappa t^{n\alpha}}{\Gamma\left( n\alpha + 1\right)}\right) \right);  
\end{equation}
which gives
\begin{equation}
X_{n + 1}\left( t\right) =  = \frac{\left( -1\right)^{\left( n+1\right)}\theta^{\left( n+1\right)}\kappa I_{0^{+}}^{\alpha}\left( t^{n\alpha}\right)}{\Gamma\left( n\alpha + 1\right)}.
\end{equation}

Thus,
\begin{equation}
I_{0^{+}}^{\alpha}\left( t^{n\alpha}\right) = \frac{1}{\Gamma\left( \alpha\right)}\int_{0}^{1}\left( t - st\right)^{\alpha-1}\left( st\right)^{n\alpha} tds = \frac{t^{\left( n+1\right)}}{\Gamma\left( \alpha\right)}\int_{0}^{1}\left( 1 - s\right)^{\alpha-1} s^{n\alpha} ds.
\end{equation}
This is equivalent to :
\begin{equation}
I_{0^{+}}^{\alpha}\left( t^{n\alpha}\right)  =  \frac{1}{\Gamma\left( \alpha\right)}\beta \left( \alpha, n\alpha + 1\right) = \frac{t^{\left( n+1\right)}\alpha \Gamma\left( n\alpha + 1\right)}{\Gamma\left( \left( n+1\right) \alpha +1\right) }
\end{equation}
and consequently
\begin{equation}
X_{n+1}\left( t\right) = \frac{\left( -1\right)^{n+1}\theta^{n+1}\kappa}{\Gamma\left( n\alpha + 1\right)}\frac{t^{\left( n+1\right)\alpha}\Gamma\left( n\alpha + 1\right)}{\Gamma\left( \left( n+1\right)\alpha + 1\right)} = \frac{\left( -1\right)^{n+1}\theta^{n+1}\kappa t^{\left( n+1\right)\alpha}}{\Gamma\left( \left( n+1\right)\alpha + 1\right)}.
\end{equation}
Thus the property is also true to order $n+1.$ So, it is true to any order $n$, i.e.
\begin{equation}
X_{n}\left( t\right) = \frac{\left( -1\right)^{n}\theta^{n}\kappa t^{n\alpha}}{\Gamma\left( n\alpha + 1\right)}.
\end{equation}
Finally, one has 
\begin{equation}
X\left( t\right)  = \displaystyle \sum_{n=0}^{+\infty}  X_{n}\left( t\right) = \displaystyle \sum_{n=0}^{+\infty}\frac{\left( -1\right)^{n}\theta^{n}\kappa t^{n\alpha}}{\Gamma\left( n\alpha + 1\right)}.
\end{equation}
\end{proof}

\begin{remark}
By considering $\alpha = \frac{1}{2}$ we obtain the formula getting by $\left(\ref{lig}\right).$ 
\end{remark}
The paper is closed with a comparison and simulation regarding the new model FDE and existing one KDEM, extensively commented on in the introduction.
\subsection{Comparative study and simulation}
\label{sec:headings3}
The main of this section is to extend previous theoretical investigate result $\left( \ref{propo}\right)$ graphically.
The comparison of the two models, the existing one and the one we propose, is obtained by application to dioxins. The elimination rate is $ \theta $ = 0.006418 and assume that the initial body burden at $ T_{0}= 0 $ is $ X\left( 0\right) = 15 pg/kg. $

\begin{table}[htbp!]
	\centering
	\caption{Evolution of the quantity of the contaminant over time}
	\begin{tabular}{|c|c|c|}
		
		\hline
		Time (months) & $X_{i+1}\left( t\right)  = X_{i}.e^{-\theta t} $ & $
		X_{i+1}\left( t\right) = X_{i}\left( e^{\theta^{2} t} - \frac{4\mid \theta \mid\sqrt{t}}{\sqrt{\pi}}\displaystyle\sum_{n = 0}^{+\infty}\frac{4^{n}\left( n + 1\right) ! \left( \theta^{2}t\right) ^{n}}{\left( 2n + 2 \right) !}\right) 
		$   \\
		\hline
		0 & 15 & 15 \\
		\hline
		1 & 14,99903733 & 14,99891376   \\
		\hline
		2 & 14,99711218 & 14,99737774\\
		\hline
		3 & 14,99422491 & 14,99549674\\
		\hline
	\end{tabular}
	\label{tab:table1}
\end{table}

The table \ref{tab:table1} shows that our proposed process results in a slightly higher amount of contaminant over time compared to the existing model. This proves that this model is well adapted to failing immune systems. Thus, it is well suited to model the dynamics of contaminant evolution in the body of people living with comorbidity, fragile children and pregnant women. Since the KDEM model is adapted to the non-failing immune system, we therefore implemented it for a graphical view of the evolution of the contaminant dynamics in the immunocompetent organism.

\begin{figure}[htbp!]
	\centering
	\includegraphics[scale=0.7]{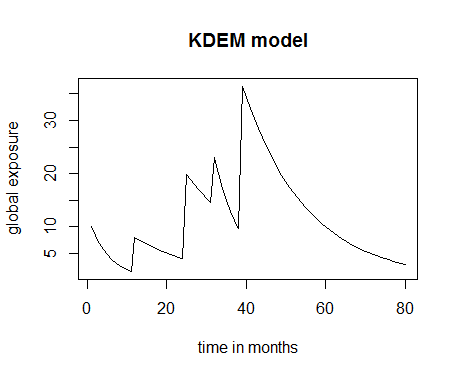}
	\caption{A trajectory of the exposure process to a contaminant through time $ \theta $ = 0.006418 }
	\label{fig:fig2}
\end{figure}

The graph above $\left(see fig \ref{fig:fig2}\right)$ shows the evolution of the dynamics of dioxins once introduced in the human organism. The value of the elimination rate (0.006418) is obtained by the formula $ \theta=ln(2)/DV $ (DV the biological half-life of dioxin is 9 years or 108 months). Given that the tolerable threshold of dioxin by the human organism is 70pg/kg of body weight, we assume however the initial body load at 10 pg/kg of body weight. We note that in this process the time taken to exceed the tolerable threshold is long, in contrast to the trajectory of contaminants in the body of people living with a comorbidity, where this threshold is exceeded more quickly after the contamination $\left( see fig \ref{fig:fig1}\right).$ We recall that this exceeding can prove to be dangerous for the organism. 

\begin{figure}[htbp!]
	\centering
	\includegraphics[scale=0.7]{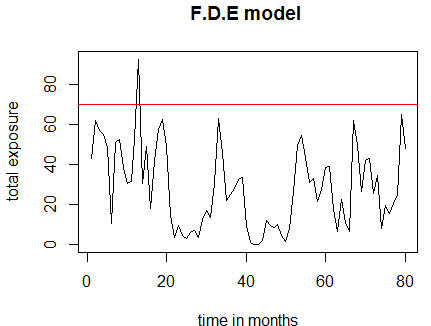}
	\caption{A trajectory of the exposure process to a contaminant through time $ \theta $ = 0.006418 }
	\label{fig:fig1}
\end{figure}

This graphical test shows that the organism of people living with comorbidity, people with weak immune system, namely the child not breastfed or malnourished, pregnant women, all these people, the system eliminates the contaminant slowly, which exposes the danger of developing a pathology more quickly than immunocompetent subjects.

\section{Conclusion and discussions}
\label{sec:headings4}
We have developed a model that takes into account a weak immune system for modeling the risks associated with poor nutrition. This allows medical services to properly assess and prevent the danger of exposure to food toxins and improve the life expectancy of these people already weakened by a pathology such as asthma, diabetes, cardiovascular diseases, AIDS... and especially malnourished children, people with disabilities, and people who are not well nourished, people of advanced age and pregnant women. We have shown that the toxin values in the organism successively obtained by the FDE model are slightly higher than the amount obtained by the existing model. Naturally, it is clear that our model is better than the KDEM model, but this confirms hypothesis (C1).\\
This fractional differential equation model that we propose in this paper to quantify dietary risk exposure marks progress in the analysis and search for solutions for the prevention of the above mentioned diseases. In our next research it will be taken into account in this new FDE model the multiple contaminations to which the organism is exposed daily, especially the application to contaminants found in food consumption in sub-Saharan African countries.

\section*{Data Availability}
 No physical data were used to support this study. The trajectories of the modeling processes are obtained by numerical simulation.
\section*{Conflicts of Interest}
The authors declare that there are no conflicts of interest for the publication of this paper.
\section*{Acknowledgements}
The authors would like to thank the Academic Editor and all the reviewers of this article for their comments and suggestions. Thank you for all that you do to advance science.

\end{document}